\documentclass{LMCS}

\def\dOi{10(4:18)2014}
\lmcsheading%
{\dOi}
{1--23}
{}
{}
{Dec.~23, 2013}
{Dec.~27, 2014}
{}

\ACMCCS{[{\bf Theory of computation}]: ???}

\newcommand{\save}[1]{}

\newcommand{\xfig}[2]{\medskip \centerline{\epsfig{figure=#1.eps,height=#2}} \bigskip}

\newcommand{\ul}[1]{\underline{#1}}

\newcommand{\mywhite}[1]{\White{#1}}

\newcommand{\outt}[1]{}
\newcommand{\delete}[1]{}
\newcommand{\blank}[1]{}
\newcommand{\notedomission}[1]{\medskip\noindent{\bf TEXT OMITTED}\\[2mm]}

\newcommand{\bxit}[1]{\hbox{\it #1}}

\newcommand{\bxrm}[1]{\hbox{\rm #1}}

\newcommand{\bxbf}[1]{\hbox{\bf #1}}

\newcommand{\bxsc}[1]{\hbox{\sc #1}}
\newcommand{\bxtt}[1]{\hbox{$\tt #1$}}

\newcommand{\fnrm}[1]{\hbox{\footnotesize #1}}

\newcommand{\scriptbf}[1]{\hbox{\scriptsize\bf #1}}

\newcommand{\scriptrm}[1]{\hbox{\scriptsize\rm #1}}

\newcommand{\bfc}{{\bf c}}
\newcommand{\bfd}{{\bf d}}

\newcommand{\bff}{{\bf f}}
\newcommand{\bfg}{{\bf g}}

\newcommand{\bfq}{{\bf q}}
\newcommand{\bfr}{{\bf r}}
\newcommand{\bfs}{{\bf s}}
\newcommand{\bft}{{\bf t}}

\newcommand{\bfv}{{\bf v}}

\newcommand{\ttd}{\hbox{$\tt d$}}
\newcommand{\tte}{\hbox{$\tt e$}}
\newcommand{\ttf}{\hbox{$\tt f$}}

\newcommand{\tti}{\hbox{$\tt i$}}

\newcommand{\tto}{\hbox{$\tt o$}}
\newcommand{\ttp}{\hbox{$\tt p$}}

\newcommand{\tts}{\hbox{$\tt s$}}
\newcommand{\ttt}{\hbox{$\tt t$}}
\newcommand{\ttu}{\hbox{$\tt u$}}
\newcommand{\ttv}{\hbox{$\tt v$}}

\newcommand{\ttzero}{\hbox{$\tt 0$}}

\newcommand{\ttone}{\hbox{$\tt 1$}}

\newcommand{\bfD}{{\bf D}}
\newcommand{\bfE}{{\bf E}}

\newcommand{\bfG}{{\bf G}}

\newcommand{\bfQ}{{\bf Q}}

\newcommand{\bfT}{{\bf T}}

\newcommand{\feq}{\hbox{\boldmath $\,\doteq\,$}}

\newcommand{\calA}{\hbox{$\cal A$}}

\newcommand{\calC}{\hbox{$\cal C$}}
\newcommand{\calD}{\hbox{$\cal D$}}
\newcommand{\calE}{\hbox{$\cal E$}}
\newcommand{\calF}{\hbox{$\cal F$}}

\newcommand{\calH}{\hbox{$\cal H$}}

\newcommand{\calM}{\hbox{$\cal M$}}
\newcommand{\calN}{\hbox{$\cal N$}}

\newcommand{\calQ}{\hbox{$\cal Q$}}

\newcommand{\calS}{\hbox{$\cal S$}}

\newcommand{\calX}{\hbox{$\cal X$}}

\newcommand{\ttB}{\hbox{$\tt B$}}

\newcommand{\ttD}{\hbox{$\tt D$}}
\newcommand{\ttE}{\hbox{$\tt E$}}
\newcommand{\ttF}{\hbox{$\tt F$}}
\newcommand{\ttG}{\hbox{$\tt G$}}

\newcommand{\ttL}{\hbox{$\tt L$}}

\newcommand{\ttN}{\hbox{$\tt N$}}

\newcommand{\ttS}{\hbox{$\tt S$}}
\newcommand{\ttT}{\hbox{$\tt T$}}

\newcommand{\ttW}{\hbox{$\tt W$}}

\newcommand{\ttZ}{\hbox{$\tt Z$}}

\newcommand{\dB}{\hbox{$\Bbb B$}}

\newcommand{\dN}{\hbox{$\Bbb N$}}

\newcommand{\dW}{\hbox{$\Bbb W$}}

\newcommand{\gra}{\hbox{$\alpha$}}

\newcommand{\grd}{\hbox{$\delta$}}
\newcommand{\gre}{\hbox{$\varepsilon$}}
\newcommand{\grh}{\hbox{$\eta$}}

\newcommand{\grl}{\hbox{$\lambda$}}
\newcommand{\grm}{\hbox{$\mu$}}
\newcommand{\grn}{\hbox{$\nu$}}

\newcommand{\grp}{\hbox{$\pi$}}

\newcommand{\grs}{\hbox{$\sigma$}}
\newcommand{\grt}{\hbox{$\tau$}}

\newcommand{\grf}{\hbox{$\varphi$}}

\newcommand{\grq}{\hbox{$\psi$}}
\newcommand{\grw}{\hbox{$\omega$}}

\newcommand{\grG}{\hbox{$\Gamma$}}
\newcommand{\grD}{\hbox{$\Delta$}}

\newcommand{\grP}{\hbox{$\Pi$}}
\newcommand{\grS}{\hbox{$\Sigma$}}

\newcommand{\bgrx}{\hbox{\boldmath$\xi$}}

\newcommand{\ra}{\rightarrow}
\newcommand{\pa}{\rightharpoonup}
\newcommand{\pra}{\rightharpoonup}	
\newcommand{\sra}{\!\rightarrow\!} 

\newcommand{\splus}{\!+\!}
\newcommand{\sminus}{\!-\!}

\newcommand{\lng}{\langle}
\newcommand{\rng}{\rangle}
\newcommand{\df}{=_{\rm df}}

\newcommand{\rsem}{]\hspace{-0.5mm}]} 
\newcommand{\lsem}{[\hspace{-0.5mm}[} 

\newcommand{\ignore}[1]{}

\newcommand{\mx}{\makebox}

\newcommand{\zero}{\rule{0mm}{3mm}}

\newcommand{\onemm}{\mx[1mm]{}}

\newcommand{\bc}{\begin{center}}
\newcommand{\ec}{\end{center}}
\newcommand{\beq}{\begin{equation}}
\newcommand{\eeq}{\end{equation}}

\newcommand{\be}{\begin{enumerate}}
\newcommand{\ee}{\end{enumerate}}
\newcommand{\bi}{\begin{itemize}}
\newcommand{\ei}{\end{itemize}}
\newcommand{\bd}{\begin{description}}
\newcommand{\ed}{\end{description}}
\newcommand{\beqn}{\begin{equation}}
\newcommand{\eeqn}{\end{equation}}

\newcommand{\beqna}{\begin{eqnarray}}
\newcommand{\eeqna}{\end{eqnarray}}
\newcommand{\beqnas}{\begin{eqnarray*}}
\newcommand{\eeqnas}{\end{eqnarray*}}
\newcommand{\beqnaa}{$$\begin{array}{rcll}}  
\newcommand{\eeqnaa}{\end{array}$$}  
\newcommand{\beqnana}{$$\begin{array}{lrcll}}  
\newcommand{\eeqnana}{\end{array}$$}  
\newcommand{\btbl}[1]{\begin{center}\begin{tabular}{#1}}
\newcommand{\etbl}{\end{tabular}\end{center}}
\newcommand{\beqnc}{$$\begin{array}{rclcl}}
\newcommand{\eeqnc}{\end{array}$$}
\newcommand{\fn}{\footnote}

\makeatletter
\def\thmlabel#1{\@bsphack\if@filesw {\let\thepage\relax
\xdef\@gtempa{\write\@auxout{\string
\newlabel{#1}{{\@Roman{\@currentlabel}}{\thepage}}}}}\@gtempa
\if@nobreak \ifvmode\nobreak\fi\fi\fi\@esphack}
\makeatother

\newcommand{\bsl}{\begin{verse}\sl}
\newcommand{\esl}{\end{verse}}

\newcommand{\einference}[2]  
  {\shortstack
      {$ #1 $\\ \mbox{}\\ $ #2 $}}

\newlength{\txtlth}
\newlength{\txtht}

\usepackage{amsmath}
\usepackage{amsmath,amssymb,oldlfont}
\usepackage{proof}
\usepackage{epsfig}
\usepackage{colordvi}
\newcommand{\savetext}[1]{}

\newcommand{\const}{\bxit{Const}}
\newcommand{\decomp}{\bxit{Deconst}}
\newcommand{\deconst}{\bxit{Deconst}}

\newcommand{\IT}{\bxbf{IT}}

\newcommand{\Rset}{\hbox{$H_{\cal C}$}}
\newcommand{\Rstruc}{\hbox{$\calH_{\cal C}$}}

\newcommand{\Dh}{\hbox{$\grD_\eta$}}

\newcommand{\ntp}{{\hbox{\tiny $\Box$}}} 

\keywords{Inductive and coinductive types, equational programs, intrinsic theories, global model theory}

\ACMCCS{[{\bf Theory of computation}]: Models of computation; Logic; Semantics and reasoning}

\usepackage{hyperref}

\begin{document}

\title[Semantic coinductive typing]{Global semantic typing\\ for inductive and 
coinductive computing}

\author[D.~Leivant]{Daniel Leivant}
\address{Indiana University Bloomington}
 \email{leivant@indiana.edu}

\begin{abstract}
Inductive and coinductive types are commonly construed as
ontological (Church-style) types, with canonical semantical interpretation.
When studying programs in the context of global
(``uninterpreted") semantics, it is preferable to think of types as
semantical properties (Curry-style).
A purely logical framework for reasoning about semantic types
is provided by intrinsic theories, introduced by the author in 2002,
which fit tightly with syntactic,
semantic, and proof theoretic fundamentals of formal logic,
with potential applications in implicit computational complexity
as well as extraction of programs from proofs.

Intrinsic theories have been considered so far for inductive data,
and we presently extend that framework to data defined using both
inductive and coinductive closures.
Our first main result is a Canonicity Theorem, showing that
the global definition of program typing, via the usual (Tarskian)
semantics of first-order logic,
agrees with their operational semantics in the intended (``canonical") model.

The paper's other main result is a proof theoretic calibration of
intrinsic theories: every intrinsic theory is interpretable in
(a conservative extension of) first-order arithmetic.
This means that quantification over infinite data objects does not lead,
on its own, to proof-theoretic strength beyond that of Peano Arithmetic.
\end{abstract}

\maketitle

\section{Introduction}

\subsection{A motivation: termination of equational programs}

We refer to the well-known dichotomy between the 
{\em canonical} and {\em global} interpretations
of proofs and programs,
often referred to as ``interpreted" and ``uninterpreted," respectively.
The former is exemplified by Peano's Arithmetic, whose canonical
model is the standard structure of the natural numbers with basic operations,
and by programming languages with primitive types for integers, strings, etc.
Thus, the axioms of Peano's Arithmetic (PA) 
are intended to contribute to
the delineation of a particular model,
whereas the axioms of Group Theory are 
intended to describe a class of models, a task they perform successfully
by definition.

The limitative properties of canonical axiomatization and computing,
e.g.\ the high complexity of program termination in the canonical model,
let alone the complexity of semantic truth of first-order formulas
of arithmetic, justify a reconsideration
of canonically intended theories, such as PA, as global theories with
unintended, ``non-standard," models.  Such non-standard models have
``non-standard elements," but the machinery of Tarskian semantics
makes no syntactic distinction between intended and non-standard elements,
and consequently no explicit distinction between
canonical and non-standard models.

A trivial remedy is to enrich the vocabulary with type-identifiers.
Indeed, that is precisely
Peano's original axiomatization of arithmetic \cite{Peano89}:
his context is an abstract universe of objects and sets, and
the natural numbers form a particular collection $N$ within that broader universe.
The type $N$ is thus construed semantically, as a collection of 
pre-existing objects, which happen to satisfy certain properties.
This is in perfect agreement with the brand of
typing introduced by Haskell Curry \cite{Curry36,Seldin-Curry}:
a function $f$ has type $\grt\sra\grs$ if it maps objects of type 
\grt\ to objects of type \grs; $f$ may well be defined for
input values that are not of type \grt.

Semantic types reflect a global perspective, in that they can be
considered for any domain of discourse.
In contrast, Church's approach \cite{Church40} construes
types as inherent properties of objects:
a function is of type $\grt \sra\grs$ when its domain consists of
the objects of \grt, and its codomain of objects
of type \grs.
That is, Church's types are related to the presence of a canonical model.\fn{The difference 
between semantic and ontological typing disciplines is thus
significant in ways that phrases such as ``explicit" and ``implicit" do not convey.}

The distinction between ontological and semantic typing can be made
for arbitrary inductive data types $D$,
such as the booleans, strings, finite trees, and lists of natural numbers.
Each inductive data-type is contained in the term algebra
generated by a set \calC\ of constructors, 
a syntactic representation that suggests a global semantics
for such types.
Namely, $D$ is a ``global predicate," that assigns to
each structure \calS\ (for a vocabulary containing \calC)
the set of denotations of closed \calC-terms.
Such global semantics is a well-known organizing 
principle for descriptive and computational devices over a class 
of structures, such as all finite graphs 
\cite{Ebbinghaus-Flum-FMT,Gurevich-challenge}.

The global viewpoint of types is of particular interest with respect to
{\em programs} over inductive data.
Each such program $P$ may be of type $D \sra D$ in one structure
and not in another; e.g.\ if $P$ is non-total on \dN, then it
is of type $\ttN \ra \ttN$ in the flat-domain structure
with \ttN\ interpreted as $\dN_\bot$, but not when \ttN\ is
interpreted as \dN.

In \cite{Leivant-intrinsic}
we showed that a program $P$ computes a total function in the canonical 
structure
iff $P$ has a unique solution, with respect to Tarskian semantics, 
in {\em every} reasonable model of $P$.
See \S\ref{subsec:inductive-canonicity} below for background and discussion.

Within the global framework, it makes
sense to consider formal theories for proving global
typing properties of equational programs. 
We focus here on equational programs, since these mesh directly 
with formal reasoning:
a program's equations can be construed as axioms,
computations as derivations in equational logic, and types as formulas.
Moreover, equational programs are amenable to
term-model constructions, which turn out to be a 
useful meta-mathematical tool.
Theories for reasoning directly about equational programs 
were developed in \cite{Leivant-intrinsic},
where they were dubbed {\em intrinsic theories.}
Among other benefits, they support
attractive proof-theoretic characterizations of major complexity
classes, such as the provable functions of Peano Arithmetic and the
primitive recursive functions \cite{Leivant-intrinsic,Leivant-unipolar}.

The rest of the paper is organized as follows.
In \S 2 we define
{\em data-systems,} i.e.\ collections of data-types
obtained by both inductive and coinductive
definitions.  Starting with the syntactic framework,
which generalizes term algebras to potentially-infinite {\em hyper-terms,}
we give an operational semantics for equational
programs over hyper-terms.
Section 3 describes the equational programs we wish to focus on, 
including their semantics.  Section 4 presents a proof of our first
main result, a {\em Canonicity Theorem} \ref{thm:canon} 
matching the Tarskian semantics of equational
programs with their operational semantics.
Section 5 describes intrinsic theories, a simple proof theoretic 
setting for reasoning about equational and typing properties of
equational programs.  Finally, \S 6 presents our other 
main theorem \ref{thm:faithful}, stating that intrinsic
theories, even in the presence of coinductive types, are interpretable
in a conservative extension of Peano's Arithmetic, and are thus of the
same proof theoretic strength as Peano's Arithmetic.\fn{The reader 
familiar with rich type systems, such as those of Coq,
Agda, or Nuprl, will notice that Theorem \ref{thm:faithful}
is stated for a type system without infinite-branching type-constructors,
such as W-types.}

\section{Data systems}

\subsection{Symbolic data}

A {\em constructor-vocabulary} is a finite set \calC\ of function
identifiers, referred to as {\em constructors,} each assigned an {\em arity}
$\geq 0$; as usual, constructors of arity 0 are {\em object-identifiers}.
Given a constructor-vocabulary \calC,
a {\em hyper-term (over \calC)} is an ordered tree 
of constructors, possibly infinite, where each node with
constructor \bfc\ of arity $r$ has exactly $r$ children.
We write \Rset\ for the set of hyper-terms over \calC.
For a structure \calS, we write $|\calS|$ for the universe of the
structure.
 
The {\em replete \calC-structure} is the structure \Rstruc\ with\fn{We use
typewriter font for actual identifiers, boldface for meta-level variables
ranging over syntactic objects, and italics for other meta-level variables.}
\be
\item
\zero\calC\ as vocabulary;
\item
\zero $|\Rstruc| = \Rset$; and
\item
a syntactic interpretation for each identifier $\bfc \in \calC$:
$\lsem\bfc\rsem(a_1 \ldots a_r)$ \quad is the tree with \bfc\ at the root 
and $a_1 \ldots a_r$ as immediate sub-trees.
\ee

\subsection{Inductive data systems}
An inductive type is defined by its generative closure
rules.  For example, the rules for \dN\ are $\ttN(\ttzero)$
and $\forall x \; \ttN(x) \ra \ttN(\tts(x))$
(we'll often omit the universal quantifier in the statements of
such rules).
Similarly, words in $\{0,1\}^*$, construed as terms generated from 
a nullary constructor \tte\ and unary \ttzero\ and \ttone,
are generated by the three rules
$\ttW(\tte)$, 
$\forall x \; \ttW(x) \ra  \ttW(\ttzero(x))$,
and $\ttW(x) \ra  \ttW(\ttone(x))$.
If \ttG\ names a type $G$,
then the type $T$ of binary trees with leaves in $G$
is generated by the rules $\ttG(x) \ra \ttT(x)$, and $\ttT(x) \wedge \ttT(y) \ra
\ttT(\ttp(x,y))$, where \ttp\ is a binary constructor (for pairing).

Several types can be generated jointly 
(i.e.\ simultaneously), for example:
the set of \bxtt{01}-strings with no adjacent \ttone's is obtained
by defining jointly the set (denoted by \ttE) of such strings that 
start with \ttone, and the set (denoted by \ttZ) of those that don't:
$\ttZ(\tte)$,
$\ttZ(x) \ra \ttZ(\ttzero(x))$,
$\ttZ(x) \ra \ttE(\ttone(x))$,
and 
$\ttE(x) \ra \ttZ(\ttzero(x))$.

Generally, a 
{\em definition of inductive types from given types $\vec{\bfG}$}
consists of:
\be
\item
A sequence $\vec{\bfD}= (\bfD_1 \ldots \bfD_k)$ 
of unary relation-identifiers, dubbed {\em type identifiers};
\item
A set of {\em construction rules,} each one of the form
\begin{equation}\label{eq:induction-rules}
\forall \vec{y} \; \bfQ_1(y_1) \wedge 
	\cdot\cdot\cdot \wedge \bfQ_r(y_r) \; \ra \;
	\bfD_i(\bfc(y_1 \cdots y_r))
\end{equation}
where \bfc\ is a constructor of arity $r$, and each $\bfQ_\ell$ is one of 
the type-identifiers in $\vec{\bfG}, \vec{\bfD}$.
\ee
These rules delineate the intended meaning of the inductive types 
$\vec{\bfD}$ from below,
as $\bfD_i$ is built up by the construction rules.

Conjoining the composition rules, we obtain a single.
The following variant, equivalent to that conjunction in
constructive (intuitionistic) first-order logic, will be useful:

\begin{equation}\label{eq:induction-packaged}
\grq_1 \vee \cdot \cdot \cdot \vee \grq_k  \;\ra\;  \bfD_i(x) 
\end{equation}
where each $\grq_i$, with $x$ a free variable, is of the form
\begin{equation}\label{eq:constructor-statement}
	\exists y_1 \ldots y_r \;\; x = \bfc(\vec{y})
	\; \wedge \; \bfQ_1(y_1) \; \wedge \cdot\cdot\cdot \wedge \bfQ_r(y_r)
\end{equation}
where $y_1 \ldots y_r$ are distinct variables.
We call a formula of the form (\ref{eq:constructor-statement})
a {\em constructor-statement (for $x$)}.

To focus on the essentials, we do not consider
several important type constructions,
such as parametric types, dependent types, sum types, polymorphism,
and W-types.

\subsection{Coinductive deconstruction rules}

Inductive construction rules state sufficient 
reasons for asserting that a 
(finite) hyper-term has a given type,
given the types of its immediate sub-terms.
The intended semantics of an inductive type $D$ is thus
the smallest set of hyper-terms closed under those rules.
Coinductive deconstruction rules state necessary conditions for a term 
to have a given type, by implying possible combinations
for the types of its immediate sub-terms.
The intended semantics is the largest set of hyper-terms satisfying 
those conditions.

For instance, the type of \grw-words over 0/1 is given by the 
deconstruction rule
\begin{equation}\label{eq:decompose}
\ttW^\omega(x) \ra 
	(\exists y \; \ttW^\omega(y) \wedge x=\ttzero (y) )
	\vee
	(\exists y \; \ttW^\omega(y) \wedge x=\ttone (y) )
\end{equation}
Note that this is not quite captured by the implications \quad
$\ttW^\omega(\ttzero x) \ra \ttW^\omega(x)$ \quad and \quad
$\ttW^\omega(\ttone x) \ra \ttW^\omega(x)$,
since these do not guarantee that every element
of $\ttW^\omega$ is of one of the two forms considered.

Moreover, using a destructor function in stating deconstruction
rules fails to differentiate between cases of the argument's
main constructor.   
For example, in analogy to the inductive definition above
of the words with no adjacent \ttone's,
the \grw-words over 0/1 with no adjacent \ttone's
are delineated jointly by the two deconstruction rules
$$
\ttZ(x) \; \ra \;
	 (\exists y \; \ttZ(y) \wedge x=\ttzero(y))
	\; \vee \;
	 (\exists y \; \ttE(y) \wedge x=\ttzero(y))
$$
and
$$
\ttE(x) \; \ra \; \exists y \; \ttZ(y) \wedge x=\ttone(y)
$$
These rules cannot be captured using a destructor, since the latter
does not differentiate between cases for the input's main constructor.

These observations motivate the following.
\begin{defi}
A {\em deconstruction definition of coinductive types from
given types $\vec{\bfG}$} consists of:
\be
\item
A list $\vec{\bfD}$ of {\em type identifiers};
\item
for each of the types $\bfD_i$ in $\vec{\bfD}$ a 
{\em deconstruction rule,} of the form 
\begin{equation}\label{eq:coinduction}
\bfD_i(x) \; \ra \; \grq_1 \vee \cdot \cdot \cdot \vee \grq_k
\end{equation}
where each $\grq_i$ is a constructor-statement (as in (\ref{eq:constructor-statement}) above).
\ee
\end{defi}

\subsection{General data-systems}
We proceed to define data-systems, in which data-types may 
be defined by any sequential nesting of induction and coinduction.
Descriptive and deductive tools for such definitions have been 
studied extensively,
e.g.\ referring to typed lambda calculi, with operators
\grm\ for smallest fixpoint and \grn\ for greatest fixpoint.  
The Common Algebraic Specification Language 
\bxsc{Casl} was used as a unifying standard in the algebraic 
specification community, and extended to coalgebraic data
\cite{Reichel99,RotheTJ01,MossakowskiSRR06,Schroder08}.
Several frameworks combining inductive and coinductive data,
such as \cite{Padawitz00}, strive to be comprehensive, including
various syntactic distinctions and categories, in contrast to our
minimalistic approach.  

\begin{defi}
A {\em data-system} \calD\ over a constructor vocabulary \calC\ 
consists of:
\be
\item
A double-list $\vec{\bfD}_1 \ldots \vec{\bfD}_k$ (the order matters)
of unary relation-identifiers, dubbed {\em type-identifiers}, where
each $\vec{\bfD}_i$ is a {\em type-bundle,}
and designated as either {\em inductive} or {\em coinductive.}
\item
For each inductive bundle $\vec{\bfD}_{i}$, an inductive definition
of $\, \vec{\bfD}_{i}$ from the types in $\vec{\bfD}_{j}$, $j< i$.

\item
For each coinductive bundle $\vec{\bfD}_{i}$ 
a {\em coinductive definition} of $\vec{\bfD}_{i}$ from 
$\vec{\bfD}_{j}$, $j< i$.
\ee
\end{defi}

\begin{defi}\label{dfn:type-rank}
We say that a data-system $\vec{\bfD}_1 \ldots \vec{\bfD}_k$ 
is $\grS_n$ ($\grP_n$) if $\vec{\bfD}_k$ is inductive 
(respectively, coinductive), and the list of bundles alternates
$n\sminus 1$ times between inductive and coinductive bundles.
(This choice of notation will become evident in Theorem
\ref{thm:type-dfn}.)
That is, a single bundle is $\grS_1$ ($\grP_1$) if it is inductive 
(respectively, coinductive); 
if $\vec{\bfD}_1 \ldots \vec{\bfD}_k$ is $\grS_n$ then
$\vec{\bfD}_1, \ldots, \vec{\bfD}_k, \vec{\bfD}_{k+1}$ 
is $\grS_n$ if $\vec{\bfD}_{k+1}$ 
is inductive, and
$\grP_{n+1}$ if $\vec{\bfD}_{k+1}$ is coinductive;
if $\vec{\bfD}_1 \ldots \vec{\bfD}_k$ is $\grP_n$ then
$\vec{\bfD}_1, \ldots, \vec{\bfD}_k, \vec{\bfD}_{k+1}$ 
is $\grP_n$ if $\vec{\bfD}_{k+1}$ 
is coinductive, and
$\grS_{n+1}$ if $\vec{\bfD}_{k+1}$ is inductive.

A data system $\calD = \vec{\bfD}_1 \ldots \vec{\bfD}_k$  
has {\em rank $n$} if it is $\grS_n$ or $\grP_n$.
A data-type $\bfD_{ij}$ of \calD\ has rank $n$ (in \calD) 
if the data-system $\vec{\bfD}_1 \ldots \vec{\bfD}_i$ has rank $n$. 
\end{defi}

\subsection{Examples of data-systems}

\be
\item
Let \calC\ consist of the identifiers
\ttzero, \ttone, \tte, \tts, and \ttp, of arities 0,0,0,1, and 2,
respectively.
Consider the following $\grS_3$ data-system, for the double list 
$((\ttB),(\ttN),(\ttF,\ttS),(\ttL))$ with inductive \ttB\ and \ttN\
(booleans and natural numbers), coinductive \ttF and \ttS\
(streams with alternating \ttB's and \ttN's starting with \ttB\ or,
respectively, \ttN),
and finally an inductive \ttL\ for lists of such streams.
The defining formulas are, in simplified form,
\[\begin{array}{c}
\ttB(\ttzero) \quad \ttB(\ttone) \\[2mm]
\ttN(\ttzero) \quad
\forall y \; \ttN(y) \ra \ttN(\tts (y)) \\[2mm]
\ttF(x) \ra \exists y,z \; (x= \ttp (y, z))
	\;\wedge\; \ttB(y)\; \wedge\; \ttS(z) \\[1mm]
\qquad \ttS(x) \ra \exists y,z \; (x= \ttp( y, z) )
	\;\wedge\; \ttN(y)\; \wedge\; \ttF(z) \\[2mm]
\ttL(\tte) \quad
\forall y,z  \; \ttF(y) \wedge \ttL(z) \ra \ttL(\ttp (y, z)) \quad
\forall y,z \; \ttS(y) \wedge \ttL(z) \ra \ttL(\ttp (y, z))
\end{array}
\]

\noindent Note that constructors \ttp\ and \ttzero\ are reused for different data-types.
This is in agreement with our untyped, generic approach, where the
data-objects are untyped.

\item\label{ex:DT}
Let the constructors be \ttzero, \ttone, \tts, \ttp, and \ttd,
of arities 0,0,1,2 and 3 respectively.
Consider the $\grP_2$ data system $((\ttN),(\ttT),(\ttD))$, 
with inductive \ttN\ (natural numbers), coinductive \ttT\ (finite
or infinite 2-3 trees with leaves in \ttN), 
and coinductive \ttD\ (infinite binary 
trees with nodes decorated by elements of \ttT).
The inductive definition of \ttN\ is as above;
the coinductive definitions of \ttT\ and \ttD\ are
$$
\begin{array}{lcl}
\ttT(x) & \ra & \ttN(x) \\
	&& \quad \vee \; (\exists y_1,y_2 \;\; x=\ttp(y_1,y_2) \; 
		\wedge \; \ttT(y_1) \; \wedge \; \ttT(y_2))\\
	&& \quad \vee \; (\exists y_1,y_2,y_3 \;\; x=\ttd(y_1,y_2,y_3) \; 
		\wedge \; \ttT(y_1) \; \wedge \; \ttT(y_2) \; 
		\wedge \; \ttT(y_3))
\end{array}
$$
and
$$
\begin{array}{lcl}
\ttD(x) & \ra & \exists u,y_1,y_2 \;\; x=\ttd(u,y_1,y_2)  \;\;
	\wedge \;\; \ttT(u) \;\; \wedge \;\; \ttD(y_1) \;\; 
	\wedge \;\; \ttD(y_2)
\end{array}
$$
Note that we construe a ``tree of trees" not as a higher-order object, 
but simply as a tree of constructors, suitably parsed.
\ee

\section{Programs over data-systems}

\subsection{Equational programs}

In addition to the set \calC\ of constructors we posit
an infinite set \calX\ of {\em variables,} and an infinite 
set \calF\ of function-identifiers, dubbed {\em program-functions,} and 
assigned arities $\geq 0$ as well. The sets \calC, \calX\ and \calF\ are,
of course, disjoint.  
If \calE\ is a set consisting
of function-identifiers and (possibly) variables, we write $\bar{\calE}$
for the set of terms generated from \calE\ by application:
if $\bfg \in \calE$ is a function-identifier of arity $r$, and 
$\bft_1 \ldots \bft_r$
are terms, then so is $\bfg\,\bft_1\, \cdots\, \bft_r$.  
We use informally the
parenthesized notation $\bfg(\bft_1, \ldots, \bft_r)$, when 
convenient.\footnote{Note that if \bfg\ is nullary,
it is itself a term,
whereas with formal parentheses we'd have $\bfg()$.}
We refer to elements of $\overline{\calC}$, 
$\overline{\calC\cup\calX}$ and
$\overline{\calC\cup\calX\cup\calF}$ as {\em data-terms,}
{\em base-terms,} and {\em program-terms,} respectively.\footnote{Data-terms 
are often referred to as {\em values}, and base-terms as {\em patterns.}}
We write $|\bft|$ for the height of a term $\bft$.

We adopt equational programs,
in the style of Herbrand-G\"{o}del, as computation model.
See for example \cite{Kleene52} for a classical exposition.
Such programs are sometimes dubbed ``computation rules"
\cite{BergerES03,SchwichtenbergW12}.
There are easy inter-translations between equational programs and
program-terms such as those of $\bxbf{FLR}_0$ \cite{Moschovakis89}.
We prefer however to focus on equational programs because 
they integrate easily into logical calculi, 
and are naturally construed as axioms.
In fact, codifying equations by terms is 
a conceptual detour, since the computational behavior of 
such terms is itself spelled out using equations or rewrite-rules.

A {\em program-equation} is an equation of the
form $\bff(\bft_1, \ldots, \bft_k) = \bfq$, 
where \bff\ is a program-function 
of arity $k \geq 0$, $\bft_1 \ldots \bft_k$ is a list of
base-terms with no variable repeating,
and $\bfq$ is a program-term.
Two program-equations are {\em compatible} if their left-hand 
sides cannot be unified.
A {\em program-body} is a finite set of pairwise-compatible 
program-equations. A {\em program} $(P,\bff)$ (of arity $k$) 
consists of a program-body $P$ and 
a program-function \bff\ (of arity $k$) dubbed the program's 
{\em principal-function.}
We identify each program with its program-body when in no danger 
of confusion.
Given a program $P$, we call the program-terms that use the
function-identifiers occurring in $P$ {\em $P$-terms.}

The requirement that program-equations have no repeating variable
in the input is essential when the input may be infinite, for else
the applicability of such an equation might depend on two inputs
being identical, a condition which is not decidable.

Programs of arity 0 can be used to define objects.  For example, the
singleton program $T$ consisting of the equation
$\ttf = \tts\tts\tts\ttzero$ defines 3, in the sense that in every
model \calS\  of $T$
the interpretation of the identifier \ttt\ is the same
as that of the numeral for 3. 
We can similarly construct nullary programs defining
hyper-terms, such as the program $I$ consisting 
of the single equation $\tti = \tts(\tti)$.
The infinite hyper-term $\tts^\omega$ is the unique hyper-term 
solution of this equation.
But that uniqueness does not extend to arbitrary structures, of course.
For example, we may have \tts\ interpreted as the identity function, and
the equation above is modeled over the ordinals
with \tts\ interpreted as $\grl x. 1+x$, and \tti\ as any infinite ordinal.

\subsection{Operational semantics of programs}

A program $(P,\bff)$
{\em computes} a partial-function \quad
$g: \, \bar{\calC} \pa \bar{\calC}$ \quad 
when $g(p)=q$ iff the equation $\bff(p)=q$
is derivable from $P$ in equational logic.
However, replete structures have infinite terms,
so the output of a program over \Rset\ must be computed
piecemeal from finite information about the input values.

To formally describe computation over infinite data, with
a modicum of syntactic machinery, we posit
that each program over \calC\
has defining equations for destructors and a discriminator.
That is, if the given vocabulary's constructors are
$\bfc_1 \ldots \bfc_k$, with $m$ their
maximal arity, then the program-functions include
the unary identifiers $\grp_{i,m}$ ($i=1..m$) and \grd\
(destructors and discriminator),
and all programs contain, for each constructor \bfc\ of arity $r$
the equations
\begin{equation}\label{eq:destr}
\begin{array}{rcll}
\grp_{i,m}(\bfc(x_1, \ldots, x_r)) &=& x_i & (i=1..r)\\
\grp_{i,m}(\bfc(x_1, \ldots, x_r))
        &=& \bfc(x_1, \ldots, x_r) \quad & (i=r\splus 1.. m)\\[1mm]
\grd(\bfc_i(\vec{\bft}), x_1, \ldots, x_k) &=& x_i  & (i = 1..k)
\end{array}
\end{equation}
We call a repeated composition of destructors a {\em deep-destructor,}
and construe it as an address in hyper-terms.

A {\em valuation} is a function \grh\ from a finite set
of variables to \Rset.
If $\vec{\bfv}$ is a list of $r$ distinct variables, and
$\vec{t}$ a list of $r$ hyper-terms, then
$[\vec{\bfv}\leftarrow \vec{t}]$ is the valuation \grh\ defined
by $\grh(\bfv_i) = t_i$.

We posit the presence in \calC\ of at least
one nullary constructor \tto;  indeed, adding a nullary constructor 
to \calC\ does not impact the rest of the discussion.
For a constructor \bfc\ we write $\bfc^o$ for the
term $\bfc(\tto,\ldots, \tto)$.
For a deep-destructor \grP\ we define\fn{Here again
we stipulate that $\calC = \{ {\bf c}_1 , \ldots {\bf c}_k\}$;
also $r_i = \fnrm{arity}({\bf c}_i)$}
$$
\grP^o(x) = \grd(\grP(x),\bfc_1^o, \ldots , \bfc_k^o)
$$
That is, $\grP^o(x)$ identifies the constructor of $x$ at address \grP.

\begin{defi}\label{dfn:comput}
We say that a set \grG\ of equations {\em locally infer}
an equation $\bft = \bfq$ between program-terms if,
for each deep-destructor \grP,
the equation $\grP^o(\bft) = \grP^o(\bfq)$
is derivable in equational logic from $\grG$.
We write then $\grG \vdash^\omega \bft = \bfq$.

\medskip

The {\em diagram} of a valuation \grh\ is the set \Dh\ of
equations of the form $\grP^o(\ttv)=\bfc^o$
where \ttv\ is in the domain of \grh, 
\grP\ a deep-destructor, and $\bfc$ the main
constructor of $\grP(\grh(\ttv))$.
That is, \Dh\ conveys, node by node, the structure of the
hyper-term $\grh(\ttv)$.

\medskip

An $r$-ary program $(P,\bff)$ 
{\em locally-computes} a partial-function \\
$g: \, \Rset^r \pa \Rset$
\quad when, for every $\vec{t}\in \Rset^r$ and $q \in \Rset$,
$g(t_1,\ldots , t_r) = q$ \quad iff \quad
$P,\, \Dh \vdash^\omega \bff(\ttv_1 \ldots \ttv_r) = \ttu$,
where $\grh = [\vec{\ttv},\ttu \leftarrow \vec{t},q]$.
\end{defi}

The notion of local-computability is motivated
solely by the presence of infinite data. For finite hyper-terms
local-computability is equivalent to computability, as we now show.
For a data-term $\bft$ of \calC\ let $\hat{\bft}$ be
corresponding hyper-term, i.e.\ the syntax-tree of \bft.
\enlargethispage{\baselineskip}

\begin{prop}\label{prop:finite-comput}
Let \bft\ and \bfq\ be data-terms.
\begin{equation}\label{1}
P, \, \grD_{[\scriptrm{u},\scriptrm{v} \leftarrow \hat{\scriptbf{t}},\hat{\scriptbf{q}}]}
	\vdash^\omega \bff(\ttv) = \ttu
\end{equation}
iff
\begin{equation}\label{2}
P \vdash \bff(\bft) = \bfq
\end{equation}
\end{prop}
By structural induction on data-terms we have, in equational
logic, and using the defining equations for the destructors
and discriminator (\ref{eq:destr}),
$$
\ttu=\bft, \, \ttv=\bfq \vdash \grD_{[\scriptrm{u},\scriptrm{v} \leftarrow \hat{\scriptbf{t}},\hat{\scriptbf{q}}]}
$$
So (\ref{1}) implies
$$
P, \, \ttu=\bft, \ttv=\bfq \vdash^\omega \bff(\ttv) = \ttu
$$
i.e.
$$
P \vdash^\omega \bff(\bft) = \bfq
$$
By induction on \bfq, and using again (\ref{eq:destr}), this implies
(\ref{2}).

For the converse, assume (\ref{2}).
Using induction on the length of
equational derivations, for all terms $\bfr,\bfs$,
if $\;P \vdash \{\bft,\bfq/\ttu,\ttv\} \, (\bfr = \bfs)\;$ then 
$$
P,\, \grD_{[\scriptrm{u},\scriptrm{v} \leftarrow \hat{\scriptbf{t}},\hat{\scriptbf{q}}]}
	\vdash^\omega  \bfr=\bfs
$$
In particular, we conclude (\ref{1}).

\subsection{Equational vs.\ Turing computation} 

The equivalence of equational programs over \dN\ with the 
\grm-recursive functions was implicit already in \cite{Godel34},
and explicit in \cite{Kleene36}.
Their equivalence with \grl-definability \cite{Church36,Kleene36a}
and hence with Turing computability \cite{Turing37} followed readily.
When equational programs are used over infinite data, a match with Turing machines must
be based on an adequate representation of infinite data by 
functions over inductive data.  
For instance, each infinite 0/1 word $w$ can be identified with 
the function \quad $\hat{w}: \; \dN\ \ra \dB$
\quad defined by $\hat{w}(k) =$ the $k$'th constructor of $w$.
Similarly, infinite binary trees with nodes decorated with 0/1 
can be identified with functions from 
$\dW = \{\ttzero,\ttone\}^*$ to $\{\ttzero,\ttone\}$.
Conversely, a function $f: \; \dN \ra \dB$ can be identified with
the \grw-word $\check{f}$ whose $n$'th entry is $f(n)$.

It follows that a functional $g:\; (\dN\ra \dB) \ra  (\dN \ra \dB)$
can be identified with the function $\check{g}: \; \dB^\omega \ra \dB^\omega$,
defined by $\check{g}(w) = (g(\hat{w}))^{\vee}$.
Conversely, a function $h: \; \dB^\omega \ra \dB^\omega$ can be identified
with the functional \quad $\hat{h}:\; (\dN\ra \dB) \ra  (\dN \ra \dB)$ \quad
defined by \quad $\hat{h}(f)=(h(\check{f}))^\wedge$. 

It is easy (albeit tedious) to see that
a partial function $h: \; \dB^\omega \ra \dB^\omega$ 
is computable by an equational
program iff the functional $\hat{h}$ is computable by 
some oracle Turing machine.
Dually, a functional $g: \; (\dN\ra \dB) \ra  (\dN \ra \dB)$ is computable
by an oracle Turing machine iff the function $\check{g}$ is computable by an
equational program.

\section{Matching Tarskian semantics and operational semantics}

\subsection{\calD-correct structures}

We have focused so far on the canonical setup for data-systems \calD, with
hyper-terms as objects.  We now consider arbitrary structures.
We call a structure \calS\ a {\em \calD-structure} if its 
vocabulary (i.e.\ signature, symbol set) contains
the constructor- and type-identifiers of \calD.
In a \calD-structure \calS\
we may have a finite or infinite regression of
constructor-eliminations, regardless of the nature of the structure
elements.  For example, if \ttf\ is a unary constructor,
and $g = \ttf_{{\cal S}}$ is its interpretation in \calS,
we might have an element $v_0 \in |\calS|$ 
for which there is a $v_1 \in |\calS|$
with $v_0 = g(v_1)$, and more generally elements
$v_i \in |\calS|$ $(i=0,1,\ldots$) where $v_i = g(v_{i+1})$.
In general $g$ need not be injective, and so
$v_{i+1}$ need not be uniquely determined by $v_i$.

We say that a \calD-structure \calS\ is {\em \calD-correct} 
(or just {\em correct} when in no danger of confusion) if
\be
\item \zero \calS\ is {\em separated for \calC,}
that is, the interpretations in \calS\ of the constructors are
all injective and have pairwise-disjoint codomains.\fn{For \dN\ 
these are Peano's Third and Fourth Axioms.}
Note that \Rstruc\ satisfies this property.
\item
If $\vec{\bfD}_{i} = \lng \bfD_{i1} \ldots \bfD_{im} \rng$ is inductive, 
then $\lng \lsem \bfD_{i1}\rsem , \ldots, \lsem \bfD_{im}\rsem \rng $
is the minimal $m$-tuple of subsets of $|\calS|$ closed under 
the construction rules for $\vec{\bfD}_{i}$, given the sets
$\lsem \vec{\bfD}_1\rsem\; \ldots\; \lsem \vec{\bfD}_{i-1}\rsem$
\item
Dually, if $\vec{\bfD}_i$ is coinductive, then
$\lsem \vec{\bfD}_{i} \rsem$ is the largest vector of subsets of $\Rset$
closed under the deconstruction rules for $\vec{\bfD}_{i}$, given the sets
$\lsem \vec{\bfD}_1\rsem\; \ldots\; 
	\lsem \vec{\bfD}_{i-1}\rsem$.
\ee
The {\em canonical model} 
$\calA \equiv \calA_{{\cal D}} \equiv \lsem \calD\rsem$ of a data-system
\calD\ is the \calD-correct expansion of the replete structure \Rstruc.
Note that inductive and coinductive types are given their 
canonical interpretation in every \calD-correct structures, 
but such structures may have elements
that are outside all types.  Indeed, that possibility
is the motivation of intrinsic theories in the first place: one deals with 
``anomalies" of computation (divergence when an inductive output is
expected, non-productiveness when a coinductive output is expected) not by
partiality, but by allowing output which is not typed.
A single element $\bot$ denoting divergence does not suffice (see the proof in 
\cite{Leivant-intrinsic} of Theorem \ref{thm:N-canon} below).

\subsection{Decomposition in data-correct structures}

Let \calS\ be a \calD-structure, and consider an element $a$ of $|\calS|$.
A {\em \calC-decomposition of $a$} is
a finitely-branching tree $T$
of elements of $|\calS| \times \calC$ such that
\be
\item
The root of $T$ is of the form $\lng a,\bfc\rng$ with $\bfc \in \calC$;
\item
if $\lng b_i, \bfc_i\rng$ ($i=1..r$) are the children in $T$ of a node 
$\lng b, \bfc\rng$ of $T$, then
$b = \bfc_{{\cal S}}(b_1,..,b_r)$.
\ee
If $a$ has a \calC-decomposition, we say that it is {\em \calC-decomposable.}
Put differently, $a$ is \calC-decomposable iff it is in the range of a partial
mapping $\grf: \; \Rset \pa |\calS|$ that satisfies $\grf(\bfc(t_1 \ldots t_r))=
\bfc_{{\cal S}}(\grf(t_1) \ldots \grf(t_r))$.

Obviously, an element $a \in |\calS|$ may have multiple 
\calC-decompositions, and even uncountably many:
it suffices to take the structure with two elements $a,b$ and
two constant functions $\grl x.a$ and $\grl x.b$.

Recall that a \calD-structure \calS\ is separated if the interpretations
in \calS\ of the constructors $\bfc \in \calC$ are injective and
with disjoint codomains.

\begin{prop}\label{prop:separated-decomposition}
If \calS\ is a separated structure for \calC, then
each element $a \in |\calS|$ has at most one decomposition.
\end{prop}
\begin{proof}
Let $t,t'$ be decompositions of $a \in |\calS|$.
We prove by induction on $n$ that if $\lng b,\bfc\rng$ is at address \gra\
of $t$ of height $n$, then it is also at address \gra\ of $t'$.
The induction's basis and step follow outright from the assumption
that  \calS\ is separated.
\end{proof}

If $t$ is a \calC-decomposition of $a$, let
$\check{t}$ be the hyper-term obtained from $t$ by replacing 
each node $\lng b,\bfc\rng$ by \bfc.
We call $\check{t}$ a {\em constructor-decomposition} (for short,
a {\em decomposition}) of $a$.
From the proof of proposition \ref{prop:separated-decomposition}
it follows that an element $a$ of a separated structure
has at most one decomposition, which we denote (when it exists) by $\check{a}$.

\begin{prop}\label{prop:typed-decomposition}
Suppose \calS\ is a \calD-correct structure.
If  $a \in |\calS|$
has type \bfD\ in \calS,
then it has a decomposition, which has type \bfD\ in \calA.

Conversely, if $t \in \Rset$ has type \bfD\ in \calA,
then every $a \in |\calS|$  which has $t$ as decomposition,
is of type \bfD\ in \calS.
\end{prop}
\begin{proof}
We prove the Proposition by cumulative induction on the rank of \bfD\
in \calD.  Suppose the statement holds for types of rank $< n$.
For each type \bfD\ of \calD\ define
$$
A(\bfD) = 
	\{ a \in |\calS| \mid \hbox{$a$ has a decomposition, which
	is in $\bfD_{{\cal A}}$} \; \}
$$
and 
$$
S(\bfD) = 
	\{ t \in \Rset \mid \hbox{$t$ is the decomposition
	of some  $a \in \bfD_{{\cal S}}$} \; \}
$$

Suppose \bfD\ is in an inductive bundle $\vec{\bfD}_i$ of \calD. 
The sequence of sets $\lng A(\bfD_{ij}) \rng_j$ satisfies 
the inductive closure condition of $\vec{\bfD}_i$.
To see this, consider a rule of \calD\ for $\vec{\bfD}_i$, say
(w.l.o.g.)
$$
\bfD(y_1) \wedge \bfD'(y_2) \wedge \bfE(y_3) \ra \bfD(\bfc(y_1,y_2,y_3))
$$
where $\bfD'$ is another type in $\vec{\bfD}_i$ and \bfE\ is
a type of rank $< n$.
We show that
$$
(y_1 \in A(\bfD))\; \wedge \;
(y_2 \in A(\bfD')) \;\wedge\; 
(y_3 \in \bfE_{{\cal S}}) \; \ra \; \bfc(y_1,y_2,y_3) \in A(\bfD)
$$
The first two premises mean that $y_1$ and $y_2$ 
have decompositions
$\check{y}_1 \in \bfD_{{\cal A}}$
and 
$\check{y}_2 \in \bfD'_{{\cal A}}$,
and the third premise
implies that
$\check{y}_3 \in \bfE_{{\cal S}}$ by IH, 
since \bfE\ is of rank $< n$.
So the hyper-term $\bfc(\check{y}_1,\check{y}_2,\check{y}_3)$
is in $\bfD_{{\cal A}}$, since \calA\ is \calD-correct.
That hyper-term is the decomposition of
$\bfc(y_1,y_2,y_3)$, proving that the latter is in $A(\bfD)$.

Since $\lng (\bfD_{ij})_{{\cal S}}\rng_j$ 
is the smallest fixpoint of those conditions (given that \calS\ is
\calD-correct), it follows that 
$\bfD_{{\cal S}} \subseteq A(\bfD)$, i.e.\
every element of $|\calS|$ of type \bfD\ in \calS\ 
has a decomposition,
which furthermore is of type \bfD\ in \calA.

For the converse, we observe that the sequence of sets 
$\lng S(\bfD_{i,j})\rng_j$  is closed under the
inductive closure conditions of the bundle $\vec{\bfD}_i$.
To see this, consider again a rule
$$
\bfD(y_1) \wedge \bfD'(y_2) \wedge \bfE(y_3) \ra \bfD(\bfc(y_1,y_2,y_3)
$$
as above.
Assume the premise of
$$
(y_1 \in S(\bfD))\; \wedge \;
(y_2 \in S(\bfD')) \;\wedge\; 
(y_3 \in \bfE_{{\cal A}}) \; \ra \; \bfc(y_1,y_2,y_3) \in S(\bfD)
$$
The first two conjuncts mean that $y_1$ and $y_2$ 
are decompositions of some
$a_1 \in \bfD_{{\cal S}}$
and 
$a_2 \in \bfD'_{{\cal S}}$,
and the third implies, by IH, 
that $y_3$ is the decomposition of some
$a_3 \in \bfE_{{\cal S}}$.
The hyper-term $\bfc(y_1,y_2,y_3)$ 
is the decomposition of
$\bfc_{{\cal S}}(a_1,a_2,a_3)$, which is in 
$\bfD_{{\cal S}}$, since \calS\ is \calD-correct.
This concludes the case where \bfD\
is an inductive type.

Suppose now that \bfD\ is coinductive.
Then $\lng A(\bfD_{ij})\rng_j$ satisfies 
the coinductive closure condition of the bundle $\bfD_i$.
To see this, consider the rule of \calD\ for \bfD, say
$$
\bfD(x) \ra \grq_1 \vee \cdot\cdot\cdot \vee \grq_k
$$ 
where each $\grq_i$ is a constructor-statement.
Assume (w.l.o.g.) that $k=1$ and
\begin{eqnarray}\label{eq:coinductiveD}
\grq_1 \quad & \equiv & \qquad 
 	\exists y_1,y_2,y_3 \; x = \bfc(y_1,y_2,y_3)
 		\; \wedge \bfD(y_1) \wedge \bfD'(y_2) \wedge \bfE(y_3)
\end{eqnarray}
where $\bfD'$ and \bfE\ are as above.
We show that
$$
a \in A(\bfD) \;\; \ra \;\;
	(\; \exists y_1 \in A(\bfD) \; \exists y_2 \in A(\bfD')
		 \; \exists y_3 \in \bfE_{{\cal S}}  
		\quad  a = \bfc_{{\cal S}}(y_1,y_2,y_3)\; )
$$
An element $a \in A(\bfD)$ has $\check{a} \in \bfD_{{\cal A}}$.
Since \calA\ is \calD-correct, $\check{a}$
must be $\bfc(t_1,t_2,t_3)$ for some
$t_1 \in \bfD_{{\cal A}}$,
$t_2 \in \bfD'_{{\cal A}}$,
and $t_3 \in \bfE_{{\cal A}}$.
Since $\check{a}$ is the decomposition of $a$, this means
that $a = \bfc_{{\cal S}}(b_1,b_2,b_3)$, where $t_i = \check{b}_i$.
So $b_1 \in A(\bfD)$, $b_2 \in A(\bfD')$, by the definition of
the function $A$,
and  $b_3 \in \bfE_{{\cal S}}$ by IH, 
since \bfE\ is of rank $< n$.

Since \calS\ is \calD-correct, $\bfD_{{\cal S}}$ is the greatest
set closed under the closure conditions for the bundle $\vec{\bfD}$;
it therefore has $A(\bfD)$ as a subset.  That is, every element $a$ of
\calS\ whose decomposition is of type \bfD\ in \calA, is of type
\bfD\ in \calS.

For the converse, we similarly prove that
$\lng S(\bfD_{i,j})\rng_j$  is closed under the
coinductive closure conditions of the bundle $\vec{\bfD}_i$.
Suppose again that the coinductive rule for \bfD\ is 
(\ref{eq:coinductiveD}) above.
We show that for every hyper-term $t$
$$
t \in S(\bfD) \;\; \ra \;\;
	( \exists y_1 \in S(\bfD) \; \exists y_2 \in S(\bfD')
		 \; \exists y_3 \in \bfE_{{\cal S}}  
		\;  a = \bfc_{{\cal S}}(y_1,y_2,y_3))
$$
Suppose $t \in S(\bfD)$, i.e.\ $t$ is the 
deconstruction of some $a \in \bfD_{{\calS}}$.
Since \calS\ is \calD-correct, $a$ must be
$\bfc(b_1,b_2,b_3)$ for some
$b_1 \in \bfD_{{\cal S}}$,
$b_2 \in \bfD'_{{\cal S}}$,
and $b_3 \in \bfE_{{\cal S}}$.
So $t$ must be of the form $\bfc(t_1,t_2,t_3)$
where $t_1 \in S(\bfD)$, $t_2 \in S(\bfD')$, by definition of
$S(\cdots)$, and $t_3 \in \bfE_{{\cal S}}$ (by IH).

Since $\bfD_{{\cal A}}$ is the greatest subset of
\Rset\ closed under the rule for \bfD, is follows that it has
$S(\bfD)$ as a subset.  That is, if a hyper-term $t$ is the decomposition
of an element of $\bfD_{{\cal S}}$, then $t$ is of type 
\bfD\ in \calA.
\end{proof}

\begin{cor}
For any two \calD-correct structures \calS\ and \calQ, if
$a \in |\calS|$ and $b \in |\calQ|$ have the same
decomposition, then they have the same types
in \calS\ and \calQ.
\end{cor}

\subsection{Typing statements}

\begin{defi}
Given a data-system \calD\ over \calC, with
$\bfD_1, \dots, \bfD_r$ and \bfE\ among its type-identifiers,
we say that a partial function
$g: \Rset^r  \pra \Rset$ is {\em of type $(\times_{j\in J} \bfD_j) \ra \bfE$} 
if 
$a_j \in (\bfD_j)_{{\cal A}}$ ($j\in J$)
jointly imply that $g(\vec{a})$ is definable
and in $\bfE_{{\cal A}}$.

If $(P,\bff)$ is a program that computes the partial-function $g$ above, we
also say that $P$ {\em is of type} $(\times_i \bfD_i) \ra \bfE$.
\end{defi}

Note that each function, including the constructors, 
can have multiple types.
Also, a program may compute a non-total mapping over \Rset,
and still be of type $\bfD \ra \bfE$, i.e.\ compute a total function from type
\bfD\ to type \bfE.

When a (total) function $\;f: \; \Rset \ra \Rset\;$ fails to be 
of a type $\bfD \sra \bfE$ there must be some $d \in \lsem\bfD\rsem$ for which
$f(d) \not\in \lsem\bfE\rsem$.  Thus the value $f(d)$ can represent
divergence with respect to computation over $\lsem\bfD\rsem$, as for example
when $\lsem\bfD\rsem = \dN$ and $\lsem\bfE\rsem = \dN_\bot$ 
with $f(d) = \bot$.
However, to adequately capture the
computational behavior of equational programs, 
multiple representations of divergence
might be necessary; see \cite{Leivant-intrinsic} for examples and discussion. 
 
The partiality of computable functions is commonly addressed
either by allowing partial structures 
\cite{Kleene69,AstesianoBKKMST02,Mosses04},
or by considering semantic domains, with an object $\bot$ denoting divergence.
The approach here is based instead on the ``global" behavior of programs in 
all structures.

\subsection{Canonicity for inductive data}\label{subsec:inductive-canonicity}

Definition \ref{dfn:comput} provides the computational semantics of a 
program $(P,\bff)$.  But as a set of equations a
program can be construed simply as a first-order formula, namely
the conjunction of the universal closure of those equations.
As such, a program has its Tarskian semantics, referring to arbitrary structures
for the vocabulary in hand, that is the constructors and
the program-functions used in $P$. 
A model of $P$ is then just a structure that
satisfies each equation in $P$.

Herbrand proposed to define a (total) function $g$ as {\em computable}
just in case there is a program for which $g$ is the unique solution.\fn{This 
proposal was
made to G\"{o}del in personal communication, and reported in \cite{Godel34}.
A modified proposal, incorporating an operational-semantics ingredient, 
was made in \cite{Herbrand32}.}
It is rather easy to show that every computable function is indeed the
unique solution of a program.
But the converse fails.  In fact, Herbrand's definition yields precisely
the hyper-arithmetical functions \cite{Rogers67}.\fn{The first 
counter-example 
to Herbrand's proposal is probably
due to Kalmar \cite{Kalmar55}. A simple example of a program
whose unique solution is not computable was given by
Kreisel, quoted in \cite{Rogers67}.}
But Herbrand's ingenious idea to relate computability of a program
to the unicity of its solution is still in force, provided one
refers collectively to all \calD-correct structures:\fn{Of course, 
the important correction of Herbrand's equational computing
is G\"{o}del's radical change of perspective, from Tarskian semantics to
operational (rewrite rules).}

\begin{thm}\label{thm:N-canon}{\rm (Canonicity Theorem for \dN) 
\cite{Leivant-intrinsic}}
An equational program $(P,\bff)$ over \dN\ computes a total function
iff the formula $\ttN(x) \ra \ttN(\bff(x))$ is true in every 
\dN-correct model of $P$.
\end{thm}

\subsection{Canonicity for Data Systems}

We generalize Theorem \ref{thm:N-canon} to all data-systems.
Given a (unary) program $(P,\bff)$ over a data-system \calD,
and a valuation \grh,
we construct a canonical model $\calM(P,\grh)$ 
to serve as ``test-structure" for the program $P$ 
and the valuation \grh\ as input.  

We define the equivalence relation 
$\approx_{P,\eta}$ over hyper-terms
to hold between $t$ and $q$
iff $\Dh,P$ locally infer $t=q$, in the sense of Definition
\ref{dfn:comput}.
When safe, we write $\approx$ for $\approx_{P,\eta}$.

Let $\calQ(P,\grh)$ be the structure
whose universe is the quotient $\Rset/\!\!\approx$, and where
each function-identifier (constructor or program-function)
is interpreted as symbolic application: 
for an $r$-ary identifier \bff, $\bff_{{\cal Q}}$ maps equivalence 
classes $[\bft_i]_\approx$ to
$[\bff (\vec{\bft})]_\approx$.
This symbolic interpretation of the constructors
guarantees that the structure is separated for $\calC$.
Let now $\calM(P,\grh)$ be the \calD-correct expansion of $\calQ(P,\grh)$,
i.e.\ the expansion of $\calQ(P,\grh)$ to the full vocabulary of \calD,
with type-identifiers, where 
inductive types are interpreted
as the minimal subsets of \Rset\ closed under their closure conditions, and
the coinductive types
as the maximal subsets closed under their closure conditions.

\begin{lem}\label{lem:M-models-P}
$\calM(P,\eta)$ is a model of $P$.
\end{lem}
\proof
If $\bff(\vec{\bft}) = \bfq$ is an equation in $P$,
then $\bff(\vec{\bft}) \approx \bfq$ is immediate
from the definition of $\approx_{P,\eta}$.
Thus 
$$
[\bff(\vec{\bft})]_\approx = [\bfq]_\approx
$$

Also, by structural induction on terms, one easily proves that
$$
\lsem \bft \rsem_{{\cal M}(P,\eta)} = [\bft]_\approx
$$
for each term \bft, since
function-identifiers are interpreted in $\calM(P,\eta)$
symbolically.

We conclude
\[\lsem \bff(\vec{\bft})\rsem_{{\cal M}(P,\eta)} 
	= \lsem \bfq\rsem_{{\cal M}(P,\eta)}\eqno{\qEd}
\]

\begin{thm}\label{thm:canon}{\rm (Canonicity Theorem for Data Systems)}
Let \calD\ be a data-system over \calC,
and \bfD, \bfE\ two type-identifiers of \calD.
Let $(P,\bff)$ be an equational program over \calC\ computing a 
partial-function $g: \; \Rset\pa\Rset$.

The following are equivalent:
\be
\item
$g: \; \bfD_{{\cal A}} \ra  \bfE_{{\cal A}}$
\item
\zero $\bfD(x) \ra \bfE(\bff(x))$ is true
in every \calD-correct model of $P$.
\ee
The equivalence above generalizes to arities $\neq  1$.
\end{thm}

\begin{proof}
We show that (1) and (2) are also equivalent to
\be
\setcounter{enumi}{2}
\item\zero
For all valuations \grh,
$\;\; \calM(P,\grh) \models \bfD(x) \; \ra \; \bfE(\bff (x))$.
\ee

\medskip
\noindent
\ul{(1) implies (2):}
Assume (1), and let \calS\ be a \calD-correct model of $P$.
Consider an element $a \in \bfD_{{\cal S}}$.
By Proposition \ref{prop:typed-decomposition}
$a$ has a decomposition $\check{a}$.
Moreover, since \calS\ is \calD-correct, the closure conditions justifying 
$a \in \bfD_{{\cal S}}$ also justify
$\check{a} \in \bfD_{{\cal A}}$.
By (1), this implies that
$g(\check{a})\in \bfE_{{\cal A}}$.

Since $g$ is computed by $P$ we have, for each deep-destructor \grP,
that an equation
$\grP^o(\bff(\ttv)) = \bfc^o$
is derivable in equational logic from $P$ and \Dh,
where $\bfc$ is the main constructor of $\grP(g(\check{a}))$.
Since $\check{a}$ is the decomposition of $a$, all equations
\Dh\ are true in \calS.
But \calS\ is known to be a model of $P$,
so $\grP^o(\bff(\ttv)) = \bfc^o$
is true in \calS\ with \ttv\ bound to $a$.  
This being the case for every deep-destructor \grP,
it follows that $\bff_{{\cal S}}(a)$ has the same\
decomposition as $g(\check{a})$.  But $g(\check{a}) \in \bfE_{{\cal A}}$
and \calS\ is \calD-correct, 
so $\bff_{{\cal S}}(a) \in \bfE_{{\cal S}}$, proving (2).

\medskip

\noindent
\ul{(2) implies (3):} $\calM(P,\grh)$ is \calD-correct by definition.
It is a model of $P$ by Lemma \ref{lem:M-models-P}.
So (3) is a special case of (2).

\medskip

\noindent
\ul{(3) implies (1):}
Assume (3).
Consider input $a \in \bfD_{{\cal A}}$,
and let $\grh(\ttv)=a$.
The class $[\ttv]_\approx$ has then the same decomposition
as $a$, and since $\calM(P,\grh)$ is \calD-correct,
it must have type \bfD\ in $\calM(P,\grh)$,
by Proposition \ref{prop:typed-decomposition}.
By (3) it follows that 
$$
\bff_{{\cal M}(P,\eta)}([\ttv]_{\approx}) \in \
	\bfE_{{\cal M}(P,\eta)}
$$
But
$$
\bff_{{\cal M}(P,\eta)}([\ttv]_{\approx}) =
	[\bff(\ttv)]_{\approx}
$$
by definition of $\calM(P,\eta)$.
Since $g$ is computed by $(P,\bff)$,
we have $\grP^o(\bff(\ttv)) = \grP^o(g(a))$ for all
deep-destructors \grP.
So $g(a)$ has the same decomposition as $[\bff(\ttv)]_{\approx}$,
and therefore is in $\bfE_{{\cal A}}$.
\end{proof}

\section{Intrinsic theories}

\subsection{Intrinsic theories for inductive data}

{\em Intrinsic theories} for inductive data-types
were introduced in \cite{Leivant-intrinsic}.
They support
unobstructed reference to partial functions and to non-denoting terms,
common in functional and equational programming.
Each intrinsic theory is intended to be a framework for
reasoning about the typing properties of programs, including their
termination and productivity. 
In particular, declarative programs are considered as formal theories.
This departs from two longstanding approaches to reasoning about
programs and their termination, namely
programs as modal operators \cite{Segerberg77,Pratt76,HarelKT},
and programs (and their computation traces) as
explicit mathematical objects \cite{Kleene52,Kleene69}.

Let \calD\ be a data-system consisting of a single inductive 
bundle $\vec{\bfD}$.
The {\em intrinsic theory for \calD}
is a first order theory over the vocabulary of \calD,
whose axioms are 
\bi
\item The closure rules of \calD.
\item
\zero{\bf Separation axioms for \calC,} 
stating that the constructors are injective
and have pairwise-disjoint codomains.
These imply that all data-terms are distinct.
\item\zero
{\bf Inductive delineation (data-elimination, Induction)},
which mirrors the inductive closure rules.
Namely, if a vector $\vec{\grf}[x]$ of first order formulas satisfies the 
construction rules for $\vec{\bfD}$, then it contains $\vec{\bfD}$:
\begin{equation}\label{eq:induction}
\const[\vec{\grf}] \; \ra \; (\wedge_i \forall x \; \bfD_i(x) \ra \grf_i[x])
\end{equation}
where $\const[\vec{\grf}]$ is the conjunction of the construction rules for the bundle,
	with each $\bfD_i(\bft)$ replaced by $\grf_i[\bft]$.
The formulas $\vec{\grf}$ are the {\em induction-formulas} 
of the delineation.

\ei

\medskip

\noindent
{\bf Example:}
Identifying $\dW=\{\ttzero,\ttone\}^*$ with the free algebra
generated from the nullary constructor \gre\ and the unary \ttzero\
and \ttone, the intrinsic theory $\bxbf{IT}(\dW)$  has as vocabulary
these constructors and a unary type-identifier $W$. 
Here we have the
\bi
\item inductive closure rules:
$$
\begin{array}{c} \infer{W(\gre)}{} \end{array}
\quad
\begin{array}{c} \infer{W(\ttzero(\bft))}{W(\bft)} \end{array}
\quad
\begin{array}{c} \infer{W(\ttone(\bft))}{W(\bft)} \end{array}
$$
\item
and inductive-delineation:
$$
\begin{array}{c}
        \infer{\grf(\bft)}
                {W(\bft) & \grf[\gre] 
			&
                        \deduce{ \grf[\ttzero (z)]}
                            {\deduce{\cdots}
                                {\{\grf[z]\}}}
			&
                        \deduce{ \grf[\ttone (z)]}
                            {\deduce{\cdots}
                                {\{\grf[z]\}}}
			}
\end{array}
$$
\ei

\begin{defi}
A unary program  $(P,\bff)$
is {\em provably of type} $\bfD \sra \bfE$ in a theory \bfT\ if \onemm
$\bfD(x) \ra \bfE(\bff(x))$
is provable in \bfT\ from the universal closure of the 
equations in $P$.\fn{Universal closure is needed, since the logic
here is first-order, rather than equational.}
\end{defi}


\begin{thm}\label{thm:provable=PA}\ 
\bxrm{\cite{Leivant-unipolar,Leivant-intrinsic}}.
\be
\item
A function $f$ over \dN\ has a program
provably of type $\ttN \sra \ttN$ in the intrinsic theory $\IT(\dN)$
iff it is a provably-recursive function
of Peano's Arithmetic, i.e.\ a function definable using 
primitive-recursion in finite types.
\item
\zero $f$ has a program proved to be of type $\ttN \sra \ttN$ using
only formulas in which \dN\ does not occur negatively 
iff $f$ is a primitive-recursive function.
\ee
\end{thm} 

\noindent Note that this characterization of
the provable functions of PA involves no particular choice of
base functions (such as addition and multiplication). 
See \cite{Leivant-intrinsic} for examples and discussion.

\subsection{Intrinsic theories for arbitrary data-systems}

Let \calD\ be a data-system.
The {\em intrinsic theory for \calD,} denoted $\IT(\calD)$, 
is a first order theory over the vocabulary of \calD,
whose axioms are the Separation axioms,
the inductive construction rules and coinductive
deconstruction rules of \calD, as well as their duals:
\bi
\item
\zero{\bf Inductive delineation (data-elimination, Induction)}:
If a vector $\vec{\grf}[x]$ of first order formulas satisfies the 
construction rules for an inductive bundle $\vec{\bfD}$, then it
contains $\vec{\bfD}$:
$$
\const[\vec{\grf}] \; \ra \; (\wedge_i \forall x \; \bfD_i(x) \ra \grf_i[x])
$$
where $\const[\vec{\grf}]$ is the conjunction of the construction rules for the bundle,
	with each $\bfD_i(\bft)$ replaced by $\grf_i[\bft]$.
\item
\zero{\bf Coinductive delineation (data-introduction, Coinduction)}:
If a vector $\vec{\grf}[x]$ of first order formulas satisfies the 
deconstruction rule for a coinductive bundle $\vec{\bfD}$, then it
is contained in $\vec{\bfD}$:
\begin{equation}\label{eq:coinduction-scheme}
\decomp[\vec{\grf}] \; \ra \; (\wedge_i \forall x \;  \grf_i[x] \ra \bfD_i(x))
\end{equation}
where $\decomp[\vec{\grf}]$ is the conjunction of the deconstruction 
rules for the bundle, with each $\bfD_i(\bft)$ replaced by $\grf_i[\bft]$.
\ei

\noindent A characterization result, analogous to Theorem
\ref{thm:provable=PA}(2), was proved in \cite{LeivantR-corecurrence}:
A function over a coinductive type is definable using corecurrence iff
its productivity is provable using coinduction for formulas in which
type-identifiers do not occur negatively.  The proof in
\cite{LeivantR-corecurrence} is for streams, the general result will
be proved elsewhere, as well as an analog Theorem
\ref{thm:provable=PA}(1).

The phrase {\em coinduction} is often mentioned in reference to the principle
enunciated by David Park,
{\em ``To prove two processes observationally equivalent, show that 
they are bisimilar"} (see e.g \cite{Sangiorgi09}).
The phrase ``observational equivalent" is sometimes taken to mean
``equal."  Park's principle is not directly derivable in the
intrinsic theory of given coinductive types because
the equality primitive of intrinsic theories is untyped, acting
as a rewrite rule (a {\em definitional equality} in the sense of
Martin-L\"{o}f's type theory).
However, intrinsic theories do derive Park's principle for equality-in-a-type.
Consider the following program for a fresh function identifier
\bxtt{eq}. 
\beqnas
\bxtt{eq}(\bfc(x_1,\ldots,x_r),\bfc(y_1,\ldots,y_r))
        & = & \bfc(\bxtt{eq}(x_1,y_1), \ldots , \bxtt{eq}(x_r,y_r))\\
        && \mywhite{\bgrx} \qquad\quad \hbox{\bfc\ a constructor of arity $r$}\\[1mm]
\bxtt{eq}(\bfc(x_1,\ldots,x_r),\bfd(y_1,\ldots,y_t)) &=&  \bgrx
        \qquad\quad \hbox{\bfc, \bfd\ distinct constructors}\\[1mm]
\eeqnas
Here \bgrx\ is a nullary constructor, not
used in any type definition of the data system.
That is, $\bxtt{eq}$ maps two equal hyper-terms into their common value,
and maps two distinct hyper-terms into a hyper-term
containing \bgrx, which is therefore in no type.
Now define an equality relation for a coinductive type \bfD\ by 
$$
x \feq_{\scriptbf{D}}\, y \quad \equiv \quad \bfD(\bxtt{eq}(x,y))
$$
Of course, this equality is undecidable, as indeed should be the case
for infinite hyper-terms.

Given a bi-simulation between $x$ and $y$, 
the program for $\bxtt{eq}$ can be used to obtain the premises
of coinduction for the unary predicate $\grl x.\bfD(\bxtt{eq}(x,y))$
(i.e.\ with $y$ as parameter).  Our Coinduction scheme (\ref{eq:coinduction-scheme})
then implies $\bfD(\bxtt{eq}(x,y))$, i.e.\ $x \feq_{\scriptbf{D}}\, y$.

\section{Proof theoretic strength}

\subsection{Innocuous function quantification}

Our general intrinsic theories refer to infinite basic objects 
(coinductive data),
in contrast to intrinsic theories for inductive data only, 
as well as traditional arithmetical theories.  
However, their deductive machinery does not imply the existence
of any particular coinductive object, as would be the case, for example, 
in the presence of some forms of the Axiom of Choice or of 
a comprehension principle. Coinductive objects can be specified, of course,
by programs, but such programs are treated as axioms, i.e.\
assumptions.

We show next that, as a consequence, any intrinsic theory \bfT\
is interpretable in a formal theory whose proof theoretic strength is
no greater than that of Peano Arithmetic.

\newcommand{\pas}{$\bxbf{PRA\!}^*$}

We take as starting point the formalism \bxbf{PRA} 
of Primitive Recursive Arithmetic,
with function identifiers for all primitive recursive functions,
and their defining equations as axioms. In addition, we have the
Separation axioms for \dN\ (as above), and the
schema of Induction for all formulas.\fn{See e.g.\ \cite{Simpson-SSOA}
for details and related discussions.}
It is well known that \bxbf{PRA} is interpretable in Peano's Arithmetic
(where only addition and multiplication are given as functions with
their defining equations).

Let \pas\ be  \bxbf{PRA} augmented with
function variables and quantifiers over them, as well as
free variables for functionals (i.e.\ functions
from numeric functions to numeric functions.)
The set of {\em terms} is built by type-correct
explicit definition (i.e.\ composition and application)
from number-, function-, and functional-variables, 
starting with \ttzero\ and identifiers for all primitive-recursive 
functions.
The theory has as axiom schema the Principle of Explicit Definition:
for each term $\bft[\vec{x},\vec{f}]$ of the extended language,
with number variables $\vec{x}$ and function variables $\vec{f}$,
$$
\forall \vec{f} \, \exists g \; \forall \vec{x} \;\;
	g(\vec{x}) = \bft[\vec{x},\vec{f}]
$$
There are no further axioms stipulating the existence of additional functions.

The schema of Induction
applies now to all formulas in the extended language.

\begin{lem}\label{lem:conservation}
The theory \pas\ is conservative over \bxbf{PRA}.
That is, if a formula in the language of \bxbf{PRA} is provable in \pas,
then it is provable already in \bxbf{PRA}.

Consequently, \pas\ is no stronger,
proof-theoretically, than \bxbf{PA}.
\end{lem}

\begin{proof}
The proof is virtually the same as that in
\cite[Prop. 1.14, p.\ 453]{TroelstraVDalen88}, that $\bxbf{E-HA}^\omega$
is conservative over \bxbf{HA}.\fn{I am grateful to
Ulrich Kohlenbach for pointing me to that reference.}
The use of classical logic, rather than constructive (intuitionistic)
logic, makes here no difference, and \pas\ is
a sub-theory of (the classical counterpart of) $\bxbf{E-HA}^\omega$.
\end{proof}

%

\newcommand{\sap}{$\bxbf{PRA}^*$}

The main result of this section, and the second of the paper,
is the following evaluation of the proof theoretic strength of
intrinsic theories, which turns out to be surprisingly modest. 

\begin{thm}\label{thm:interpretability}{\rm (Arithmetic interpretability)}
Every intrinsic theory is interpretable in \pas.
\end{thm}

As will become clear from the proof of Theorem \ref{thm:type-dfn} below,
Theorem \ref{thm:interpretability} depends on our avoiding
infinite-branching type constructions, such as W-types.

\subsection{Representing data by numeric functions}

We posit canonical primitive-recursive coding-scheme
$\lng \cdots \rng$ for sequences of natural numbers.
More generally, we assume that
basic syntactic operations on finite data-terms, such as application and
sub-term extraction, are represented by primitive recursive
functions.
See e.g.~\cite{Kleene52,Rose84} for details, related
notations, and
proofs of the closure of the primitive-recursive functions and predicates 
under major operations, such as bounded quantification and minimization.

For each constructor \bfc, let $\bfc^\sharp$ be a distinct numeric code.
We say that a function $f: \; \dN \ra \dN$
{\em represents} a hyper-term $t \in \Rset$
if $f$ maps addresses $a= \lng a_0 \cdots a_k \rng \in \dN$ 
to the code $\bfc^\sharp$ of the constructor
\bfc\ at address $a$ of $t$, whenever such a constructor
exists.
(We could insist that $f(a)$ be some flag, say $0$,
when $t$ has no constructor at address $a$,
thereby determining $f$ uniquely from $t$; but this would be of no use
to us, and would imply the undecidability of determining whether
two computable functions represent the same {\em finite} term.)

For example, the finite term
$\ttp(\tte,\ttzero(\tte))$ is represented by $f$ provided
$f\lng\rng= \ttp^\sharp$,
$f\lng 0\rng= \tte^\sharp$,
$f\lng 1\rng= \ttzero^\sharp$, and
$f\lng 1,0\rng = \tte^\sharp$.
Similarly, the infinite
01-word  $(\ttzero,\ttone)^\omega = \ttzero\ttone\ttzero\ttone\cdots$ 
is represented by $f$
provided $f\lng 0^{2n}\rng = \ttzero^\sharp$ and 
$f\lng 0^{2n+1}\rng = \ttone^\sharp$ ($n \geq 0$).

It follows that an $r$-ary constructor \bfc\ is represented by a
functional $\bfc^\ntp$ provided
\beqnaa
\bfc^\ntp(f_1, \ldots, f_r)\lng\rng &=& \bfc^\# \\
\bfc^\ntp(f_1, \ldots, f_r)(\lng i \rng * a) &=& f_i(a) & i=1..r \\
\eeqnaa

\subsection{Representing types by formulas}

Consider a purely co-inductive data-system.
One can state that a hyper-term $t$ is of type \bfD\ by asserting
the existence of a correct type decoration of the nodes of $t$, with
the root assigned type \bfD.  The correctness of the decoration can
be expressed by a single numeric $\forall$, but we would need to have
an existential function-quantifier to state, in the first\
place, the existence of the decoration.
We show here that no such function quantification is needed.
Referring to Definition \ref{dfn:type-rank}, we have:

\begin{thm}\label{thm:type-dfn}
\zero
$\grS_n$ types are defined by $\grS_n^0$ formulas, and
$\grP_n$ types by $\grP_n^0$ formulas.
\end{thm}
\begin{proof}
The proof is by induction on $n\geq 1$.
If a type \bfD\ is $\grS_n$, i.e.\ is in a bundle defined inductively 
from $\grP_{n-1}$ types (where we take $\grP_0$ to be empty), 
then a hyper-term $t$ has type \bfD\ iff there is a finite deduction
establishing $\bfD(t)$ from typing-statements of the form $\bfE_j(t_j)$, 
with each $\bfE_j$ a type of lower rank, 
and $t_j=\grP(t)$ for some deep-destructor \grP.
By IH each $\bfE_j$ is defined by some $\grP_{n-1}^0$
formula $\bfE^\ntp[t_j]$ and the correctness of the finite type-derivation is
clearly a primitive-recursive predicate. Thus $\bfD$ is definable
by existential quantification over $\grP_{n-1}^0$ formulas, i.e.\
by a $\grS_n^0$ formula $\bfD^\ntp$.

Consider now a $\grP_n$ type \bfD.
We shall concretize the general argument using a running example,
with types $\ttD$ and $\ttE$ defined by a common coinduction, 
and \ttT\ a (previously defined) $\grS_{n-1}$ type:
\beqnas
\ttD(x) & \ra & \exists y,z \; x=\ttp(y,z) \wedge \ttD(y) \wedge \ttE(z) \\ 
& & \vee \; \exists y,z \; x=\ttp(y,z) \wedge \ttT(y) \wedge \ttD(z) \\ 
& & \vee \; \exists y \; x=\ttf(y) \wedge \ttE(y)\\[1mm]
\ttE(x) & \ra & \exists y,z \; x=\ttp(y,z) \wedge \ttE(y) \wedge \ttD(z)
\eeqnas
The decomposition rule for each type has a number of constructor-statements
as choices, in our example \ttD\ has three and \ttE\ one.
Each choice determines a main constructor, and types for the component.
The spelling out of \ttD\ into three options can be represented graphically:

\xfig{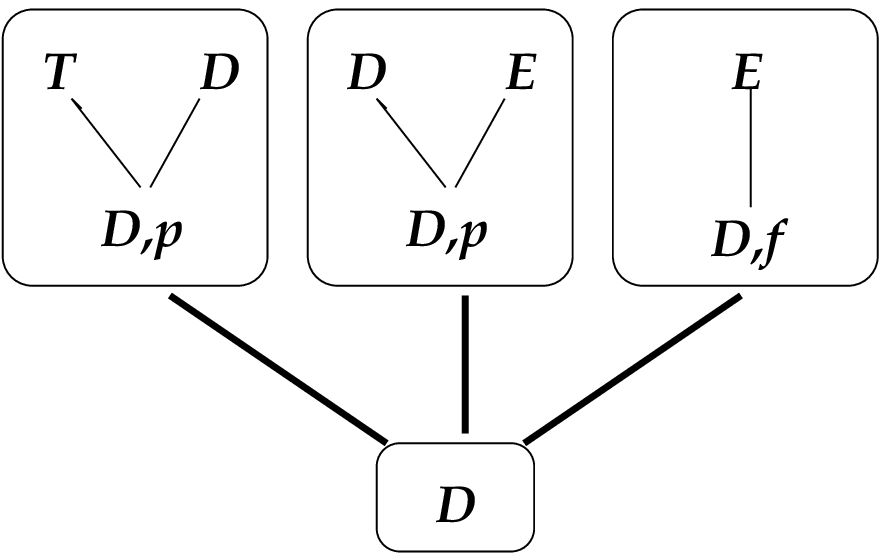}{3.2cm}

We continue an expansion of all typing options for a hyper-term
in \ttD.  That is, we construct a tree $T_D$, where a node of height $h$
consists of a finite tree (of height $\leq h$), with types at the leaves, and
a pair of a type and a constructor at internal nodes.  Each such 
node represents a possible partial typing of a 
hyper-term of type \ttD.
The children of each such node $N$ are the
local expansions of the lowermost-leftmost unexpanded leaf, with a type
in the bundle considered.  (E.g., in our running example, 
the leaves with type \ttT\ are not expandable, and are left alone.)
Put differently, the leaves are expanded
in a breadth-first order.  (We refrain from expanding all expandable leaves 
at each step, because the resulting tree, albeit finitely-branching, would
have unbounded degree.)

\xfig{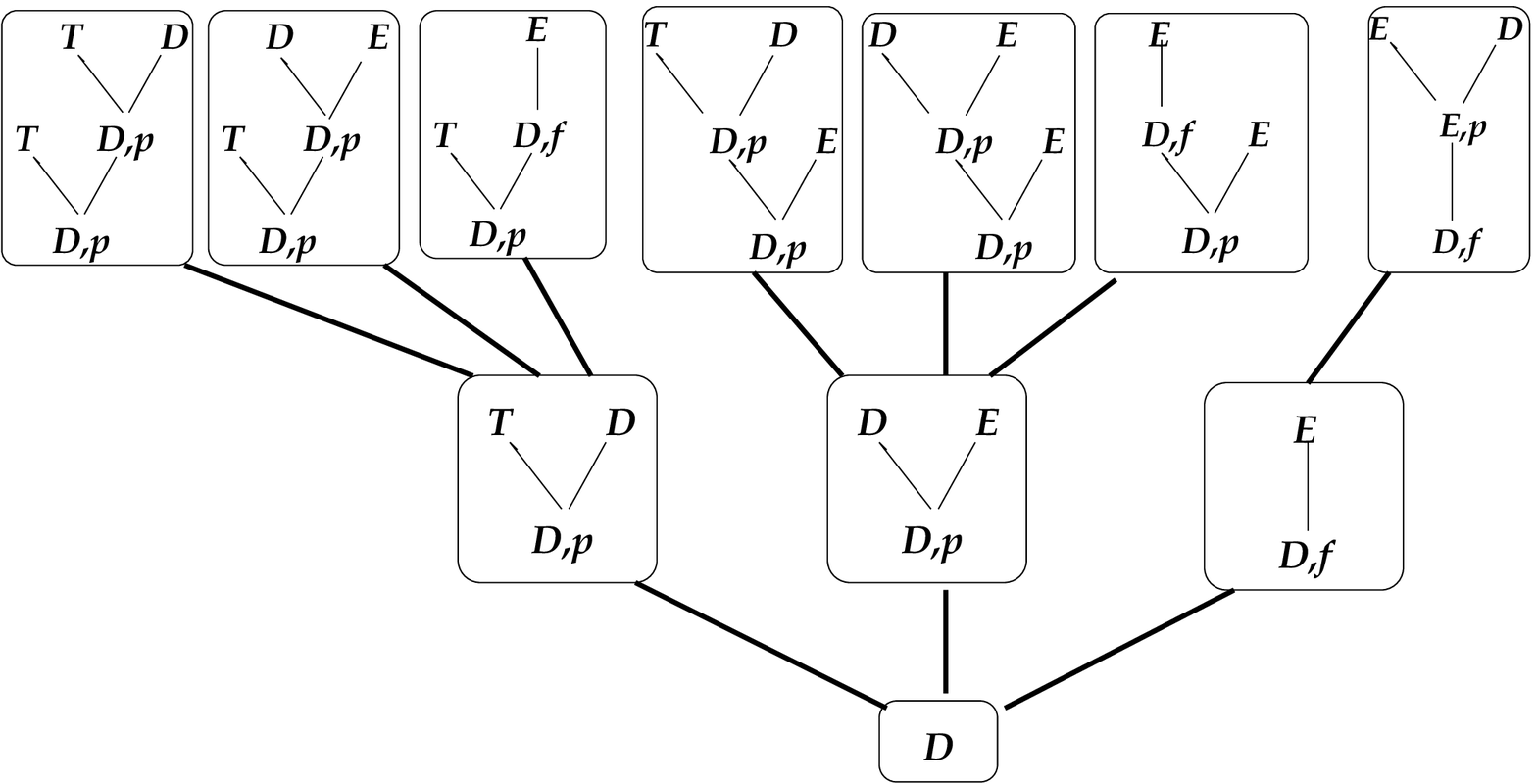}{6cm}

\noindent
A few nodes of height 3 are given here:

\xfig{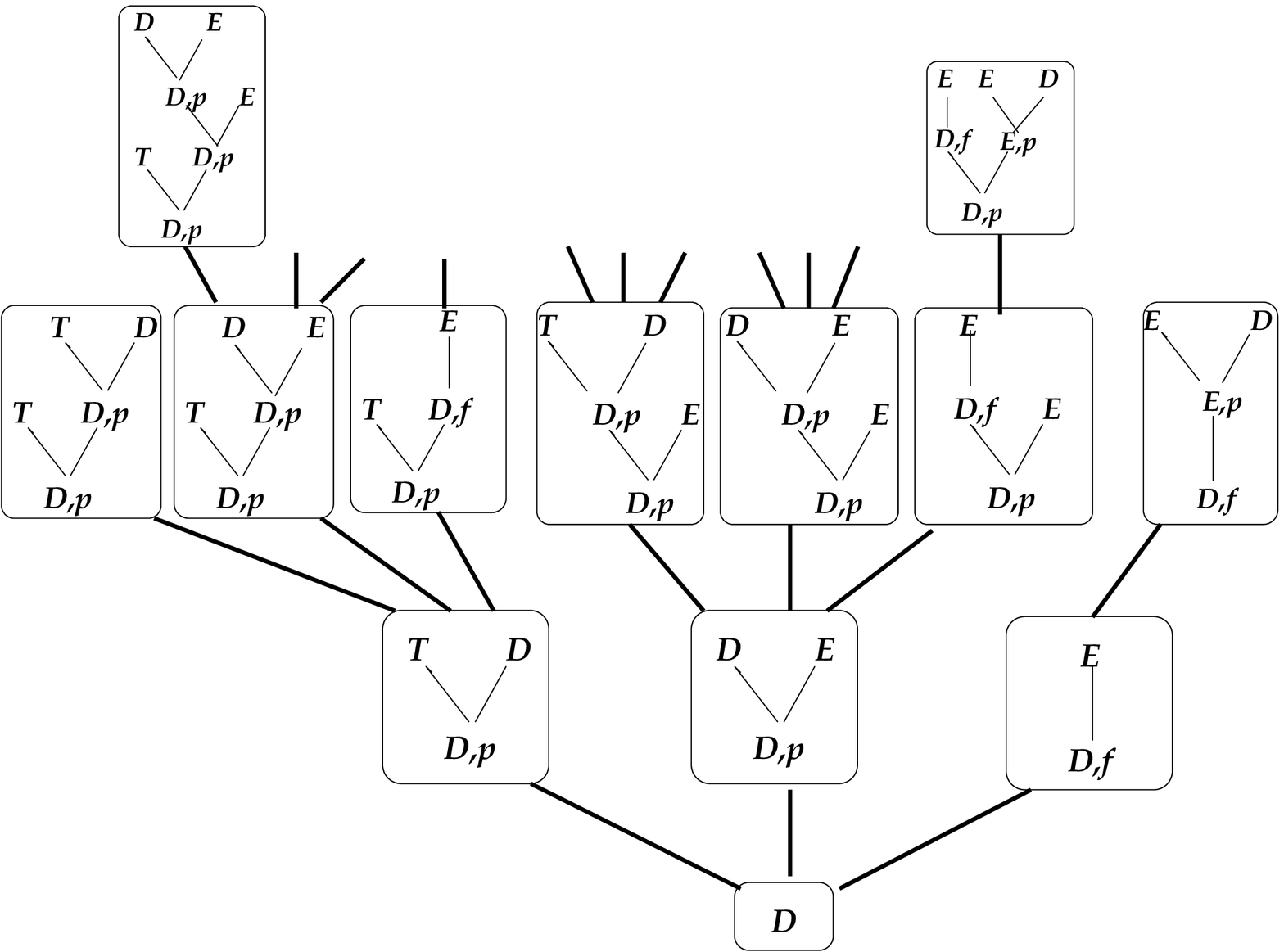}{8cm}

Note that the tree $T_D$ is primitive recursive, i.e.\ there is
a primitive-recursive function that, for every address \gra,
gives (a numeric code for) the node at address \gra.

A hyper-term $t \in \Rset$ is {\em consistent} with a node $N$ as above
if its constructor-decomposition is consistent with the 
tree of constructors in $N$, and for every deep-destructor \grP,
if $N$ has at address \grP\ a type \bfE\ of lower rank,
then $\grP(t)$ has type \bfE. The consistency of
a hyper-term $t$ with a node $N$ is thus definable by a $\grS_{n-1}^0$
formula.

A hyper-term has type \ttD\ iff
there is an unbounded (i.e.\ infinite or terminating) branch of the
tree $T$ above, every node of which is consistent with $t$.
The existence of such a branch is equivalent, by Weak K\"{o}nig's Lemma,
to the existence, for every $h > 0$, of a node $N$ of height $h$
in $T_D$, which is consistent with $t$.  Since consistency of $t$ with $N$
is definable by a $\grS_{n-1}^0$ formula, this property is $\grP_n^0$.
\end{proof}

\subsection{Interpretation of terms}

For an $r$-ary constructor \bfc\ let $\hat{\bfc}$ denote the PR functional
that maps functions $f_0,\ldots , f_{r-1}$ representing hyper-terms 
$t_0,\ldots, t_{r-1}$ to a function representing the hyper-term
obtained by rooting $t_0,\ldots , t_{r-1}$ from the symbol \bfc, i.e.
\beqnas
\hat{\bfc}(f_0,\ldots , f_{r-1})\lng\rng &= & \bfc^\sharp\\
\hat{\bfc}(f_0,\ldots , f_{r-1})\lng i \rng * a &= & f_i(a)
\eeqnas
Recall that we posit the presence in the vocabulary of \sap\
of identifiers for all PR functionals, in particular $\hat{\bfc}$.

Next, we define a mapping $\bft\mapsto \bft^\ntp$ from
terms of \bfT\ to terms of \pas.  We posit that the identifiers of \pas\
for PR functions and functionals are disjoint from the program
identifiers of intrinsic theories.

\bi
\item
For a variable $x$ of \bfT\ (intended to range over hyper-terms)
we let $x^\ntp$ be a fresh unary function variable of \pas\
(intended to range over functions representing hyper-terms).

\item
For a constructor \bfc\ of arity $r\geq 0$, let
$(\bfc(\bft_1 \ldots \bft_r))^\ntp \df \hat{\bfc}(\bft_1^\ntp \ldots \bft_r^\ntp)$.

\item
For a program-function \bff\ of arity $r \geq 0$ (i.e.\ a free variable
denoting a function between hyper-terms), let 
$(\bff(\bft_1 \ldots \bft_r))^\ntp \df 
	\bff^\ntp(\bft_1^\ntp \ldots \bft_r^\ntp)$,
where $\bff^\ntp$ is a fresh functional variable of \pas, of arity $r$.
\ei

\subsection{Interpretation of formulas}\enlargethispage{\baselineskip}

Finally, we define a mapping $\grf \mapsto \grf^\ntp$ from formulas of \bfT\
(possibly with program-functions) to formulas of \pas.
Let $\bxit{Htm}[g]$ be a PR formula stating that the function $g$
represents a hyper-term.

\bi
\item
\zero $(\bft = \bfq)^\ntp$ \quad is \quad $\bft^\ntp = \bfq^\ntp$.
\item
\zero $(\bfD(\bft))^\ntp$ \quad is \quad $\bfD^\ntp[\bft^\ntp]$, where $\bfD^\ntp$ is the
arithmetic formula (possibly with free function and functional variables)
that defines \bfD\  (Theorem \ref{thm:type-dfn}).

\item
\zero $(\grf\wedge\grq)^\ntp$ \quad is \quad $\grf^\ntp\wedge\grq^\ntp$,
and similarly for the other connectives.

\item
\zero $(\forall x \; \grf)^\ntp$ \quad is \quad 
	$\forall x^\ntp \; \bxit{Htm}[x^\ntp] \ra \grf^\ntp$; \quad
$(\exists x \; \grf)^\ntp$ \quad is \quad 
        $\exists x^\ntp \; \bxit{Htm}[x^\ntp] \wedge \grf^\ntp$.
\ei

\begin{prop}\label{prop:faithfulness}
The mapping $\grf \mapsto \grf^\ntp$ is semantically faithful; that is,
for each formula $\grf[\vec{x},\vec{\bff}]$ of \bfT, with
free object variables among $\vec{x}$  and program-variables
among $\vec{\bff}$,
$$
\calA,\; [\vec{x}\! \leftarrow\! \vec{t}, 
		\; \vec{\bff} \! \leftarrow \! \vec{g}]
	\; \models \; \grf
$$
iff for all unary functions $\vec{h}$ over \dN\ representing
(respectively) the hyper-terms $\vec{t}$, and for all functionals
$\vec{G}$ representing (respectively) the functions $\vec{g}$,
$$
\calN,\; [\vec{x}^\ntp\! \leftarrow\! \vec{h},
	\; \vec{\bff}^\ntp\! \leftarrow\! \vec{G}]
        \; \models \; \grf^\ntp
$$
In particular, if \grf\ is a closed formula of \bfT, then
\grf\ is true in the canonical
model \calA\ of the data-system iff $\grf^\ntp$ is true in the 
standard model of \pas.
\end{prop}
\begin{proof}
The proof is straightforward by structural induction on \grf.
\end{proof}

\subsection{An Interpretability Theorem}

We finally show that the interpretation is proof-theoretically
faithful.

\begin{thm}\label{thm:faithful}
If a closed formula
$\grf$ is provable in the intrinsic theory \bfT,
then $\grf^\ntp$ is provable in \pas.

More generally: if a formula $\grf[\vec{x}]$, with free
variables among $\vec{x}$,
is provable in \bfT, then
\quad $ \bxit{Htm}[\vec{x}] \ra  \grf^\ntp[\vec{x}] $ \quad
is provable in \pas.
\end{thm}
\begin{proof}
The proof proceeds by structural induction on derivations.
\begin{description}
\item[Logic]
The propositional and quantifier inferences are trivially pressured 
by the interpretation.

\item[Separation]
The case of the Separation Axioms is immediate by the definition of the
interpretation.

\item[Inductive construction]
Consider the construction axiom (\ref{eq:induction-packaged})
for an inductive bundle $\vec{\bfD}$, 
$$
\grq_1 \vee \cdot \cdot \cdot \vee \grq_k  \;\ra\;  \bfD_i(x)
$$
where each $\grq_i$ is a constructor-statement.
The interpretation of (\ref{eq:induction-packaged}) is 
$$
\bxit{Htm}[\vec{x}^\ntp] \ra
(\grq_1^\ntp \vee \cdot \cdot \cdot \vee \grq_k^\ntp  \;\ra\;  \bfD_i^\ntp[x^\ntp])
$$
with
$$
\grq_i^\ntp \quad \hbox{of the form} \quad
        \exists y_1^\ntp \ldots y_r^\ntp \in \bxit{Htm}\; \; 
	x^\ntp = \bfc^\ntp(\vec{y})
        \; \wedge \; \bfQ_1^\ntp[y_1] \; \wedge \cdot\cdot\cdot \wedge 
	\bfQ_r^\ntp[y_r]
$$
Recall (from Theorem \ref{thm:type-dfn}) that 
$\bfD_i^\ntp[x^\ntp]$ states the existence of a finite type derivation \grD\ of
$\bfD[x^\ntp]$ from statements of the form
$\bfE[\grP(x^\ntp)]$ with \bfE\ of lower rank and \grP\ a deep-destructor.  
Thus one of the decompositions $\grq_i^\ntp$ must be true for $x^\ntp$,
with the correctness of the $\bfQ_j^\ntp[y_j]$ true by induction on the
height of \grD.

\item[Induction]
Given an inductive bundle $\vec{\bfD}$,
the interpretation of $\vec{\bfD}$-induction for
formulas $\vec{\grf}$ (\ref{eq:induction}) is
\begin{equation}\label{eq:ind-interp}
\bfD_i^\ntp[x^\ntp]  \; \ra \;
(\const^\ntp[\vec{\grf}^\ntp] \; \wedge \; \bxit{Htm}[x^\ntp]
	\; \ra \; \grf_i^\ntp[x^\ntp])
\end{equation}
Recall that $\bfD_i^\ntp[x^\ntp]$ states the existence of a finite derivation
\grD\ of $\bfD(x^\ntp)$
from formulas of the form $\bfE(\grP(x^\ntp))$, where \bfE\ is of lower rank than \bfD,
and \grP\ is a deep-destructor.
The conclusion of (\ref{eq:ind-interp}) is straightforward by cumulative (i.e. course-of-value)
induction on the height of \grD.

\item[Coinductive deconstruction]
A deconstruction axiom (\ref{eq:coinduction}) for a coinductive bundle 
$\vec{\bfD}$ has the interpretation
$$
\bfD_i^\ntp[x^\ntp] \; \ra \; \grq_1^\ntp \vee \cdot \cdot \cdot \vee \grq_k^\ntp
$$
where each $\grq_i$ is a constructor-statement.
Recall that the definition of $\bfD_i^\ntp$ in this case (Theorem 
\ref{thm:type-dfn}) refers to the tree $T_D$ of expansion-options
for objects of type \bfD.  Continuing our running example in the proof
of Theorem  \ref{thm:type-dfn}, $\bfD_i^\ntp[x^\ntp]$ implies the existence,
at every height, of an expansion of \bfD\ which is consistent with
the structure of $x^\ntp$.  Nodes of height 2 which are consistent
with $x^\ntp$ give its decomposition, say as $\ttp(y^\ntp,z^\ntp)$.
One of these nodes must have above it nodes of arbitrary height
consistent with $x^\ntp$.\fn{Note that we do not use here 
Weak K\"{o}nig's Lemma,
as we do not assert the existence of an infinite branch as a consequence.}
If that node is the leftmost, giving $\ttT(y^\ntp)$ and $\ttD(z_0)$, then we
have $\ttD^\ntp[z^\ntp]$ by assumption, and $\ttT^\ntp[y^\ntp]$ holds 
by the definition of $\ttD^\ntp$ (since the node for $\ttT(y^\ntp)$ is 
a leaf of $T_D$).

\item[Coinduction]
Given a coinductive bundle $\vec{\bfD}$,
the interpretation of $\vec{\bfD}$-coinduction for
formulas $\vec{\grf}$ (\ref{eq:coinduction}) is
\begin{equation}\label{eq:coind-interp}
\grf_i^\ntp[x^\ntp] \; \ra \;
((\deconst^\ntp[\vec{\grf}^\ntp] \; \wedge \bxit{Htm}[x^\ntp])
	\; \ra \bfD_i^\ntp[x^\ntp])
\end{equation}
The conclusion of (\ref{eq:coind-interp}) is established by showing that
the tree $T_D$ (see the proof of Theorem \ref{thm:type-dfn})
has a node consistent with $x^\ntp$ at any given height $h$.
This follows outright from the assumptions of (\ref{eq:coind-interp}) by {\em induction} 
on $h$.\qedhere
\end{description}
\end{proof}

\section{Applications and further developments}

Intrinsic theories provide a minimalist framework for reasoning about
data and computation.  The benefits were already evident when dealing with
inductive data only, including a characterization of the provable functions
of Peano's Arithmetic without singling out any functions beyond 
the constructors, a particularly simple proof of Kreisel's Theorem that 
classical arithmetic is $\grP_2^0$-conservative over intuitionistic 
arithmetic \cite{Leivant-intrinsic}, and a particularly simple 
characterization of the primitive-recursive functions \cite{Leivant-unipolar}.
The latter application guided a dual characterization of the primitive
corecursive functions in terms of intrinsic theories with positive coinduction
\cite{LeivantR-corecurrence}.

Intrinsic theories are also related to type theories, 
via Curry-Howard morphisms,
providing an attractive framework for extraction of computational contents
from proofs, using functional interpretations and realizability methods.  
The natural extension of the framework to coinductive methods, 
described here, suggests new directions
in extracting such methods for coinductive data as well.
Recent work by Berger and Seisenberg \cite{BergerS12} 
has already explored similar
ideas. 

Finally, intrinsic theories are naturally amenable to ramification, 
leading to a transparent Curry-Howard link with ramified recurrence 
\cite{BellantoniC92,Leivant-ramified-rec}
as well as ramified corecurrence \cite{RamyaaL11}.

\small


\bibliographystyle{plain}
\bibliography{x}

\end{document}